\let\accentvec\vec  
\documentclass{llncs}
 
\let\vec\accentvec 

\usepackage{subfigure}
\usepackage{llncsdoc}
\usepackage{amssymb}
\usepackage{amsmath}
\usepackage{enumerate}
\usepackage{algorithmic}
\usepackage{algorithm}

\spnewtheorem{observation}{Observation}{\bfseries}{\itshape}

\newcommand{\N}{\mathbb{N}}

\newcommand{\doubleQ}{\mathbb{Q}}

\newcommand{\mzero}{\mbox{\bf{0}}}

\newcommand{\G}{\mathcal{G}}

\begin{document}

\title{Parameterized Two-Player Nash Equilibrium}

\author{Danny Hermelin \and Chien-Chung Huang \and Stefan Kratsch \and Magnus Wahlstr\"{o}m}
\institute{
Max-Planck-Institute for Informatics, Saarbr\"{u}cken, Germany\\
\email{\{hermelin,villars,skratsch,wahl\}@mpi-inf.mpg.de}
}
\date{}

\maketitle

\begin{abstract}
We study the computation of Nash equilibria in a two-player
normal form game from the perspective of parameterized
complexity. Recent results proved hardness for a
number of variants, when parameterized by the support size. We
complement those results, by identifying three cases in which
the problem becomes fixed-parameter tractable. These cases
occur in the previously studied settings of sparse games and unbalanced games as well as in the newly considered case of locally bounded treewidth games that generalizes both these two cases.

\end{abstract}

\section{Introduction}
\label{Section: Introduction}

Algorithmic game theory is a quite recent yet rapidly  developing discipline that lies at the intersection of computer science and game theory. The emergence of the internet has given rise to numerous applications in this area such as online auctions, online advertising, and search engine page ranking, where humans and computers interact with each other as selfish agents negotiating to maximize their own payoff utilities. The amount of research spent in attempting to devise computational models and algorithms for studying these types of interactions has been overwhelming in recent years; unsurprisingly perhaps, when one considers the economical rewards available in this venture.

The central problem in algorithmic game theory
is that of computing a \emph{Nash equilibrium}, a set of
strategies for each player in a given game, where no player can gain by changing his strategy when all other players strategies remain fixed. This problem is so important because Nash equilibria provide a good way to predict the outcomes of many of the scenarios described above, and other scenarios as well. Furthermore, Nash's Theorem states that for any finite game a mixed Nash equilibrium always exists. However, for this concept to be meaningful for predicting behaviors of rational agents which are in many cases computers, a natural prerequisite is for it to be computable. This led researchers such as Papadimitriou to dub the problem of computing Nash equilibria as one of the most important complexity problems of our time~\cite{Papadimitriou2001}.

The initial breakthrough in determining the complexity of computing Nash equilibria was made by Daskalakis, Goldberg, and Papadimitriou~\cite{DaskalakisGoldbergPapadimitriou2006,GoldbergPapadimitriou2005}. These two papers introduced a reduction technique which was used by the authors for showing  that computing a Nash equilibrium in a four player game is PPAD-complete. Shortly afterwards, this hardness result was simultaneously extended to three player games by Daskalakis and Papadimitriou~\cite{DaskalakisPapadimitriou2005}, and by Chen and Deng~\cite{ChenDeng2005}. The case of two player (bimatrix) games was finally cracked a year later by Chen and Deng~\cite{ChenDeng2006}, who proved it to be PPAD-complete. This implied the existence of a polynomial-time algorithm for the core case of bimatrix games to be unlikely.

Since the result of Chen and Deng~\cite{ChenDeng2006}, the
focus on computing Nash equilibria in bimatrix games was
directed either towards finding approximate Nash equilibria~\cite{BosseByrkaMarkakis2010,ChenDengTeng2006b,ChenDengTeng2006a,ChenTengValiant2007,DaskalakisMehtaPapadimitriou2007,DaskalakisMehtaPapadimitriou2009,LiptonMarkakisMehta2003}, or towards finding special cases where exact equilibria can be computed in polynomial time~\cite{Addario-BerryOlverVetta2007,ChenDengTeng2006a,CodenottiLeonciniResta2006,KalyanaramanUmans2007,LiptonMarkakisMehta2003}.
Nevertheless, for general bimatrix games the best known algorithm for computing either approximate or exact equilibria essentially tries all possibilities for the support of both players (the set of strategies played with non-zero probability), which can be assumed to be at most logarithmic in the approximate case~\cite{LiptonMarkakisMehta2003}. Once the support of both players is known, one can compute a Nash equilibrium by solving a linear-program.

\begin{theorem}[\cite{NRTV2007}]
\label{theorem:computingnash}
A Nash equilibrium in a bimatrix game, where the support sizes are bounded by $k$, can be computed in $n^{O(k)}$ time.
\end{theorem}

Due to the central role that the algorithm of Theorem~\ref{theorem:computingnash} plays in computing exact and approximate Nash equilibria, it is natural to ask whether one can improve on its running-time substantially. In particular, can we remove the dependency on the support size from the exponent? The standard framework for answering such questions is that of parameterized complexity theory~\cite{DowneyFellows1999,FlumGrohe2006}.
Estivill-Castro and Parsa initiated the study of computing Nash equilibria in this context~\cite{Estivill-CastroParsa2009} . They showed that when the support size is taken as a parameter, the problem is W[2]-hard even in certain restricted settings. The implication of their result is a negative answer to the above question. In particular, combining their reduction with the results of Chen \emph{et al.}~\cite{Chenetal2005} gives a sharp contrast to Theorem~\ref{theorem:computingnash} above.

\begin{theorem}[\cite{Estivill-CastroParsa2009}]
\label{theorem:computingnashlowerbound}
Unless \textnormal{FPT=W[1]}, there is no $n^{o(k)}$ time algorithm for computing a Nash equilibrium with support size at most $k$ in a bimatrix game.
\end{theorem}

The consequence of Theorem~\ref{theorem:computingnashlowerbound} above is devastating in the sense that for large enough games that have equilibriums with reasonably small supports, the task of computing equilibria already becomes infeasible. The main motivation of this paper is to find scenarios where one can circumvent this. Our goal is thus to identify natural parameters which govern the complexity of computing Nash equilibria, and which can help in devising feasible algorithms. We believe that this direction can prove to be fruitful in the quest for understanding the computational limitations of this fundamental problem. Indeed, prior to our work, Kalyanaraman and Umans~\cite{KalyanaramanUmans2007} provided a fixed-parameter algorithm for finding equilibrium in bimatrix games whose matrices have small rank (and some additional constraints).

Our techniques are based on considering a natural graph-theoretic representation of bimatrix games. This is done by taking the union of the underlying boolean matrix of the two given payoff matrices, and considering this matrix as the biadjacency matrix of a bipartite graph. A similar approach was taken by~\cite{CodenottiLeonciniResta2006}, and in particular by~\cite{Addario-BerryOlverVetta2007} who considered games that have an underlying planar graph structure. Our work complements both these results as will be explained further on.

A natural class of games that has a convenient interpretation in the graph-theoretic context is the class of $\ell$-sparse games~\cite{ChenDengTeng2006a,CodenottiLeonciniResta2006,DaskalakisPapadimitriou2009}.
Here each column and row in both payoff matrices of the game
have at most $\ell$ non-zero entries. An initial tempting
approach in these types of games would be to try to devise a  parameterized algorithm with $\ell$ taken as a single parameter. However, Chen, Deng, and Teng~\cite{ChenDengTeng2006a} showed that unless PPAD = P, there is no algorithm for computing an $\varepsilon$-approximate equilibrium for a 10-sparse game in time polynomial both in $\varepsilon$ and $n$. Thus, such an FPT algorithm cannot exist unless PPAD is in P. We complement this result by showing that if $\ell$ is taken as a parameter, and the size of the supports is taken as an additional parameter, then computing Nash equilibrium is fixed-parameter tractable.
\begin{theorem}
\label{theorem:sparsegame}
A Nash equilibrium in a $\ell$-sparse bimatrix game, where the support sizes is bounded by $k$, can be computed in
$\ell^{O(k\ell)} \cdot n^{O(1)}$ time.
\end{theorem}

Note that the above result also complements the polynomial-time algorithms given in~\cite{ChenDengTeng2006a,CodenottiLeonciniResta2006} for 2-sparse games. While in these algorithms there was no assumption made on the size of support of the equilibrium to be found, both algorithms could handle only \emph{win-lose} games~\cite{AbbottKaneValiant2005,ChenTengValiant2007}, games with boolean payoff matrices.

Our second result is concerned with
\emph{$k$-unbalanced games}, games where the row player has a
small set of $k$
strategies~\cite{KalyanaramanUmans2007,LiptonMarkakisMehta2003}.
Lipton, Markakis, and Mehta~\cite{LiptonMarkakisMehta2003} observed that in such games there is always an equilibrium where the row player plays a strategy with support size at most~$k+1$. Thus, by applying Theorem~\ref{theorem:computingnash} one can find a Nash equilibrium in~$n^{O(k)}$ time for these types of games. Can this result be improved to an algorithm running in~$f(k) \cdot n^{O(1)}$ time? We give a partial answer to this question, by showing that if the number~$\ell$ of different payoffs of the row player is taken as an additional parameter, the problem indeed becomes fixed-parameter tractable.
\begin{theorem}
\label{theorem:unbalancedgame}
A Nash equilibrium in a~$k$-unbalanced bimatrix game, where the row player has~$\ell$ different payoff values, can be computed
in~$\ell^{O(k^2)} \cdot n^{O(1)}$ time.
\end{theorem}

Our last result considers bimatrix games whose corresponding graph has a convenient structural property, namely the property of having locally bounded treewidth. Note that both $\ell$-sparse games and $k$-unbalanced games have corresponding graphs with this property, as well as the games with an underlying planar graph structure considered by~\cite{Addario-BerryOlverVetta2007}. We show that in games of locally bounded treewidth, where the payoff matrices have at most $\ell$ different values, one can compute a Nash equilibrium in $f(k,\ell) \cdot n^{O(1)}$. Although this might seem as a generalization of both of our results mentioned above, the reader should note that here we have a stricter requirement on the number of different values in the payoff matrices, and the running-time dependency on both parameters increases much faster.

\begin{theorem}
\label{theorem:localtreewidthgame}
A Nash equilibrium in a locally bounded treewidth game, where the support sizes are bounded by~$k$, and the payoff matrices have at most~$\ell$ different values, can be
computed computed in~$f(k,\ell) \cdot n^{O(1)}$ time for some
computable function~$f()$.
\end{theorem}

The paper is organized as follows:
We begin with some preliminaries in Section~\ref{section:preliminaries}. In Section~\ref{section:sparsegames} we consider $\ell$-sparse games and prove Theorem~\ref{theorem:sparsegame}. Section~\ref{section:localtreewidthgame} addresses locally bounded treewidth games and proves Theorem~\ref{theorem:localtreewidthgame}. Finally, in Section~\ref{section:unbalancedgames} we prove Theorem~\ref{theorem:unbalancedgame} regarding $k$-unbalanced games.

\section{Preliminaries}\label{section:preliminaries}

Let~$\mathcal{G}:=(A,B)$ be a bimatrix game, where~$A,B \in \mathbb{Q}^{n \times n}$ are the \emph{payoff matrices} of the \emph{row} and the \emph{column} players respectively. The row (column) player has a strategy space consisting of the rows (columns)~$[n]:=\{1\cdots n\}$. (For ease of notation, except in unbalanced games, we assume that both players have the same number of strategies; different numbers of strategies do not affect any of our results.) The row (column) player chooses a strategy profile~$x$ (\emph{resp.}~$y$), which is a probability distribution over his strategy space, that is,~$x_i,y_j \geq 0$,~$\forall i,j$, and furthermore~$\sum_{i=1}^{n}x_i=1$ and~$\sum_{j=1}^{n}y_j=1$. 
The expected outcomes of the game for the row and the column players are~$x^TAy$ and~$x^T By$ respectively.

The players are rational, always aiming for maximizing their expected payoffs. They have reached a \emph{Nash equilibrium} if the current strategies~$x$ and~$y$ are such that neither player has a deviating strategy~$x'$ resp.~$y'$ such that~$x'^TAy>x^TAy$ resp.~$x^T By'>x^T By$, i.e., if neither of them can improve his payoff independently of the other. The following proposition gives an equivalent condition for a pair of strategies to be an equilibrium.

\begin{lemma}\label{lemma:nashcondition}
\emph{(\cite[Chapter~3]{NRTV2007})} The pair of strategy vectors~$(x,y)$ is a Nash equilibrium for the bimatrix game~$(A,B)$ if and only if
\begin{enumerate}[(i)]
\item~$x_s > 0 \Rightarrow (Ay)_{s} \geq (Ay)_{j}$ for all~$j \neq s$;
\item~$y_s > 0 \Rightarrow (x^{T}B)_{s} \geq (x^{T}B)_{j}$ for all~$j \neq s$.
\end{enumerate}
\end{lemma}

The \emph{support} of a strategy vector~$x$ is defined as the set~$S(x)=\{i|x_i >0\}$. Note that the above proposition implies that if~$(x,y)$ is a Nash equilibrium, in the column vector~$Ay$, the entries
in~$S(x)$ are equivalent and no less than all other entries not in~$S(x)$; symmetrically, in the row vector~$x^TB$, the entries in~$S(y)$ are equivalent and no less than other entries not in~$S(y)$. It is known that, given possible supports~$I,J\subseteq[n]$ it can be efficiently decided whether there is a matching Nash equilibrium, and the corresponding strategy vectors can be computed via linear programming (Theorem~\ref{theorem:computingnash}).

The following graph associated with a bimatrix game is useful for presenting our algorithms in Sections~\ref{section:sparsegames} and~\ref{section:localtreewidthgame}.

\begin{definition} \label{def:gamegraph}
Let~$\G=(A,B)$ be a bimatrix game with~$A,B\in \doubleQ^{n \times n}$. The undirected (and bipartite) graph~$G=G(\G)$ associated with~$\G$ is defined to be the bipartite graph with vertex
classes~$V_r,V_c:=[n]$, referred to as \emph{row} resp.\ \emph{column} vertices. There is an edge between~$i \in V_r$ and~$j \in V_c$ iff~$A_{i,j}\neq 0$ or~$B_{i,j}\neq 0$.
\end{definition}

As a last bit of notation: For~$I,J \subseteq [n]$, and any~$n \times n$ matrix~$A$, we use~$A_{I,J}$ to denote the submatrix composed of rows in~$I$ and columns in~$J$. We also use~$A_{I,*}$ as a shorthand for~$A_{I,[n]}$. Thus,~$A_{i,*}$ means the~$i$'th row of~$A$.

\section{Sparse Games}
\label{section:sparsegames}

In this section we present the proof for Theorem~\ref{theorem:sparsegame}. Throughout the section we let $\mathcal{G}:=(A,B)$ denote our given bimatrix game, where $A$ and $B$ are rational value matrices with at most $\ell$ non-zero entries per row or column. We will present an algorithm for finding an Nash equilibrium where the support sizes of both players are at most~$k$ (and~$k$ is taken as a parameter).
The high-level strategy is to show that it suffices to search for equilibria that induce one or two connected components in the associated graph~$G=G(\mathcal{G})$. This permits us to find candidate support sets by enumerating subgraphs of~$G$ (on one or two components).

We begin by introducing minimal equilibria:

\begin{definition}
A Nash equilibrium~$(x,y)$ is minimal if for any Nash equilibrium~$(x',y')$
with~$S(x') \subseteq S(x)$ and~$S(y') \subseteq S(y)$, we have~$S(x')=S(x)$ and~$S(y')=S(y)$.
\end{definition}

Our algorithm iterates through all possible support sizes~$k_1,k_2 \leq k$ in increasing order to determine whether there exists an equilibrium $(x,y)$ with $|S(x)|=k_1$ and $S(y)=k_2$. To avoid cumbersome notation, we will assume that $k_1=k_2=k$ (extending this to general case will be immediate). Thus at a given iteration, the algorithm can assume that no equilibrium exists with smaller supports, \emph{i.e.} it can restrict its search to minimal equilibriums. This fact will prove crucial later on. Furthermore, since our game is $\ell$-sparse, our algorithm only needs to search for equilibriums where both player receive non-negative payoffs.

\begin{lemma}
If $\G=(A,B)$ is an $\ell$-sparse game, where $A,B \in \mathbb{Q}^{n \times n}$ and $n > \ell k$, then in any Nash equilibrium with support at most~$k \times k$, both players receive non-negative payoffs.
\end{lemma}

For an equilibrium~$(x,y)$, let the \emph{extended support} of~$x$ be the rows~$S(x) \cup N(S(y))$,
and similarly for~$y$, where the neighborhood~$N(I)$ is taken over the graph~$G:=G(\G)$ of the game.
Note that any row not in the extended support of~$x$ would have payoff constantly zero
given the current strategy of~$y$, and thus is not important for the existence of an equilibrium.
We will show that for a minimal equilibrium~$(x,y)$, the extended supports of~$x$ and~$y$ induce
a subgraph of~$G$ which has at most two connected components. This will be done in two steps: The first is the special case where~$A_{S(x),S(y)}=B_{S(x),S(y)}=\mzero$, while the second corresponds to the remaining cases.

\begin{lemma} \label{lemma:connected}
If~$(x,y)$ is a minimal Nash equilibrium for a game~$(A,B)$ with $A_{S(x),S(y)} =  B_{S(x),S(y)} = \mzero$, then the subgraphs induced by~$N[S(x)]$ and~$N[S(y)]$ in the graph associated with the game are both connected.
\end{lemma}

\begin{proof}
Let $G_x:=G[N[S(x)]]$ be the subgraph of $G$ induced by $N[S(x)]$, and aiming towards a contradiction, suppose that $G_x$ is disconnected. Let $C$ be a connected component in $G_x$, and write $p:=\sum_{i \in V_r(C)} x_i$ to denote the probability that a row strategy in $C$ is played according to $x$. Now define a new row strategy by setting~$\hat{x}_i=x_i/p$ if~$i \in V_x$, and~$\hat{x}_i=0$ otherwise.
We argue that~$(\hat{x},y)$ is a Nash equilibrium, contradicting the fact that $(x,y)$ is minimal.

Obviously, the expected payoff in~$(\hat{x},y)$ is zero for both players, as $A_{S(\hat{x}),S(y)}$ = $B_{S(\hat{x}),S(y)}$ = $\mzero$.
Furthermore, there is no row~$i$ such that~$(Ay)_i>0$, since the strategy~$y$ is unchanged and
the original strategy pair $(x,y)$ is an equilibrium. Now assume that there is a column~$j$ such that~$(\hat{x}^TB)_j>0$.
Then~$B_{i,j}\neq 0$ for some~$i \in S(\hat{x})$, and by the connectivity assumption
we must have~$B_{i,j}=0$ for all~$i \in S(x) \setminus S(\hat{x})$,
but then~$(x^TB)_j=p(\hat{x}^TB)_j>0$, which contradicts the assumption that~$(x,y)$ is a Nash equilibrium.
By Lemma~\ref{lemma:nashcondition}, it follows that~$(\hat{x},y)$ is a Nash equilibrium. \qed
\end{proof}

\begin{lemma} \label{lemma:connected2}
If~$(x,y)$ is a minimal Nash equilibrium for a game~$(A,B)$ with either $A_{S(x),S(y)} \neq \mzero$ or $B_{S(x),S(y)} \neq \mzero$, and with a non-negative payoff for both players, then the subgraph induced by~$N[S(x) \cup S(y)]$ is connected.

\end{lemma}

\begin{proof}
Let $H$ be the subgraph induced by~$N[S(x) \cup S(y)]$ 
and suppose that~$H$ is disconnected.
We will derive a contradiction by showing that $(x,y)$ is not minimal.

Let $C$ be a connected component in~$H$ intersecting both~$S(x)$ and~$S(y)$, and let $V_{\hat{x}}:=V_r(C) \cap S(x)$ and $V_{\hat{y}}:=V_c(C) \cap S(y)$ denote the set of row and column strategies in $C$, respectively.
 We define a new pair of strategy profiles $(\hat{x},\hat{y})$ where $S(\hat{x}) = V_{\hat{x}}$ and $S(\hat{y}) = V_{\hat{y}}$, by normalizing $(x,y)$ onto $V_{\hat{x}}$ and $V_{\hat{y}}$. That is, we let $p:=\sum_{i \in V_{\hat{x}}} x_i$, and set $\hat{x}_i=x_i/p$ if~$i \in V_{\hat{x}}$, and~$\hat{x}_i=0$ otherwise. Similarly, we let $q:=\sum_{j \in V_{\hat{y}}} y_i$, and set $\hat{y}$ accordingly. As either $S(\hat{x}) \subset S(x)$ or $S(\hat{y}) \subset S(y)$, to prove the lemma it suffices to argue that $(\hat{x},\hat{y})$ is an equilibrium.

Consider a row strategy $i \in [n]$. We claim
\begin{equation}
(A\hat{y})_i =
\begin{cases} (Ay)_i/q & \text{if $A_{i,V_{\hat{y}}} \neq \mzero$,}  \\
              0 &        \text{otherwise.}
\end{cases}
\label{eqn:ayhatcases}
\end{equation}
The second case is clear.  For the first case, assume that~$A_{i,V_{\hat{y}}} \neq \mzero$.
Now, if~$A_{i,j}\neq 0$ for~$j \in S(y) \setminus V_{\hat{y}}$,
then there would be an edge in~$H$ from the vertex corresponding to row~$i$, which is in~$C$,
to the vertex corresponding to column~$j$, which is not in~$C$,
contradicting that~$C$ is a connected component.
Thus
\[
(Ay)_i = \sum_{j \in S(y)} A_{i,j} y_j =
 \sum_{j \in V_{\hat{y}}} A_{i,j} (q\hat{y}_j) + \sum_{j \in S(y) \setminus V_{\hat{y}}} A_{i,j}y_j =
q(A\hat{y})_i + 0
\]
and the claim follows.

Now let $s \in S(\hat{x})$, and consider some arbitrary row strategy $i \in [n]$.
Assume by way of contradiction that~$(A\hat{y})_s < (A\hat{y})_i$.
It is clear that $A_{i,V_{\hat{y}}} \neq \mzero$, thus~$(A\hat{y})_i=(Ay)_i/q$;
we consider the cases for row~$s$.
If $A_{s,V_{\hat{y}}} \neq \mzero$, then by~(\ref{eqn:ayhatcases}), we have~$(A\hat{y})_s=(Ay)_s/q$,
implying~$(Ay)_s < (Ay)_i$.
On the other hand, if~$A_{s,V_{\hat{y}}}=\mzero$ but~$A_{s,S(y)}\neq \mzero$,
then~$C$ would not be a connected component of~$H$ (since a neighbor of~$s$ would be missing).

Thus~$A_{s,S(y)}=\mzero$ and~$(Ay)_s=(A\hat{y})_s=0$, and~$(Ay)_i=q(A\hat{y})_i>0$.
In  both cases we contradict that~$(x,y)$ is an equilibrium.
Thus we fulfill condition~(i) of Lemma~\ref{lemma:nashcondition}, and by symmetry we also fulfill condition~(ii).
We have shown that~$(\hat{x},\hat{y})$ is an equilibrium, contradicting the minimality of~$(x,y)$.
\qed \end{proof}

As an immediate corollary of Lemma~\ref{lemma:connected} and~\ref{lemma:connected2}, we get that the subgraph in $G(\G)$ induced by the extended support of a minimal equilibrium has at most two connected components.
In the following lemma we show that in graphs of small maximum degree, we can find all such subgraphs quite efficiently.
This will allow us to find a small, minimal equilibrium by checking all sets of
rows and columns that would be candidates for being the extended supports of one.

\begin{lemma}
\label{lemma:subgraph finding}
Let~$G$ be a graph on~$n$ vertices and with maximum degree~$\Delta=\Delta(G)$. In time~$(\Delta+1)^{2t}\cdot n^{c+O(1)}$ one can enumerate all subgraphs on~$t$ vertices that consist of~$c$ connected components.
\end{lemma}

\begin{proof}
We first show how to enumerate all connected subgraphs on~$t$ vertices in time~$(\Delta+1)^{2t}\cdot n^{O(1)}$ by a branching algorithm. At any point selected vertices will be active or passive. When a vertex is selected it will first be active and later be set to passive. Selecting a vertex and making it active respectively setting a vertex to passive is called an \emph{event}.

First, we branch on the choice of one out of~$n$ starting vertices in~$G$ and set it active. Then until we have selected~$t$ vertices we do the following: We consider the least recently added active vertex and branch on one of at most~$\Delta+1$ events, namely selecting one of its at most~$\Delta$ neighbors (and making it active) or setting the vertex itself to passive. We terminate when we have selected~$t$ vertices (and output the corresponding subgraph) or when there are no more active vertices. Clearly, on each branch of this algorithm there are at most~$2t$ events. Thus the branching tree has at most~$(\Delta+1)^{2t}$ leaves, implying a total runtime of~$(\Delta+1)^{2t}\cdot n^{O(1)}$.
Observe that for every connected subgraph on~$t$ vertices there is a sequence of events such that the graph occurs in the enumeration.

The generalization to~$c$ components is straightforward: When there are no more active vertices but we have not yet selected~$c$ components, then we select one of the less than~$n$ remaining vertices of~$G$ as a new active vertex, starting a new component. Selecting new starting vertices adds a factor of~$n^c$ to our runtime.
\qed\end{proof}

We are now in position to describe our entire algorithm. It first iterates through all possible sizes of extended support in increasing order. In each iteration, it enumerates all subgraphs that might correspond to the extended support of a minimal equilibrium.
It then checks all ways of selecting a support from the given subgraph, and for each such selection it uses the algorithm behind Theorem~\ref{theorem:computingnash} to check whether there is an equilibrium on the support. If no equilibrium is found throughout the whole process, the algorithm reports that there exists no equilibrium with support size at most $k$ in $\G$.
The running time is bounded by~$\ell^{O(k\ell)}n^{O(1)}$ from Lemma~\ref{lemma:subgraph finding}, times~$\binom{k\ell  + k}{k}^2=2^{O(k\ell)}$ ways of selecting the support, times~$n^{O(1)}$ for checking for an equilibrium.
In total, we get a running time of~$\ell^{O(k\ell)}n^{O(1)}$.

Finally, completeness comes from the exhaustiveness of Lemma~\ref{lemma:subgraph finding} and the structure given by Lemmas~\ref{lemma:connected} and~\ref{lemma:connected2}.

\subsection{Games with Non-negative Payoffs}

In the case that the payoffs of our games are non-negative, i.e.,~$A,B \in \doubleQ^{n \times n}_{\geq 0}$,
we can reduce our running time to be polynomial in~$\ell$, for~$\ell$-sparse games.
We begin with a strengthening of Lemmas~\ref{lemma:connected} and~\ref{lemma:connected2}.

\begin{lemma} \label{lemma:nonnegconnected}
Let~$\G=(A,B)$ be a bimatrix game with~$A,B \in \doubleQ_{\geq 0}^{n \times n}$,
and~$G$ be the graph associated with~$\G$.
If~$(x,y)$ is a minimal Nash equilibrium for~$\G$, then
either~$|S(x)|=|S(y)|=1$, or~$G[S(x) \cup S(y)]$ is connected.
\end{lemma}
\begin{proof}
Let~$G_S := G[S(x) \cup S(y)]$; assume that~$G_S$ is not connected.
If the expected outcome is zero, then (since the entries are non-negative)
every entry in~$A_{S(x),*}$ and~$B_{*,S(y)}$ is zero, and
we get an equilibrium by selecting any single row~$i \in S(x)$ and column~$j \in S(y)$.
Otherwise, every row of~$A_{S(x),S(y)}$ and every column of~$B_{S(x),S(y)}$
contains some positive entry.
Let~$C$ be a connected component of~$G_S$ on row vertices~$V_{\hat x}$ and column vertices~$V_{\hat y}$,
and define a new pair of strategy profiles $(\hat{x},\hat{y})$ where $S(\hat{x}) = V_{\hat{x}}$ and $S(\hat{y}) = V_{\hat{y}}$, by normalizing $(x,y)$ onto $V_{\hat{x}}$ and $V_{\hat{y}}$ as in Lemma~\ref{lemma:connected2}.

We will argue that~$(\hat{x},\hat{y})$ is an equilibrium.

Let~$s \in V_{\hat x}$, and assume by way of contradiction that for some row~$i\in[n]$,
we have~$(A\hat y)_{i}>(A\hat y)_s$.  Let~$(A\hat y)_s=c_0$; by non-connectivity of~$G_S$,~$(Ay)_s=qc_0$.
Further let~$(A\hat y)_{i}=c_1$ and~$(Ay)_{i}=qc_1 + (1-q)c_2$.
Now~$(Ay)_s=qc_0 < qc_1 \leq qc_1+(1-q)c_2=(Ay)_i$, contradicting
our assumptions; the last inequality is because the entries are non-negative.
Repeating the argument symmetrically, we find that~$(\hat x,\hat y)$ is a Nash equilibrium.
\qed \end{proof}

Thus, to find an equilibrium in~$\G=(A,B)$, it suffices to search for occurrences of the support,
rather than the extended support.
Invoking Lemma~\ref{lemma:subgraph finding} directly with a bound of~$2k$ vertices
gives a running time of~$\ell^{O(k)}n^{O(1)}$.

\section{Locally Bounded Treewidth Games}
\label{section:localtreewidthgame}

\newcommand{\WLGsupportsize}{k}

Let~$\G=(A,B)$ be a given game with~$A,B\in P^{n \times n}$,
with~$P \subset \doubleQ$,~$|P| \leq \ell$,
and let~$G=G(\G)$ the graph associated with~$\G$.
In this section we will present an
algorithm that finds an equilibrium with support sizes at
most~$\WLGsupportsize$ when~$G$
comes from a graph class with locally bounded treewidth.
Note that this is a partial extension of the results of the
previous section, as graphs of bounded degree have locally
bounded treewidth, while on the other hand we assume that there is a bounded set~$P$
of only~$\ell$ different payoff values which can occur in the games.
(The case~$P=\{0,1\}$ would correspond to win-lose games.)

\begin{definition}[\cite{Eppstein1999}]
A graph class has \emph{locally bounded treewidth} if
there is a function~$f:\N\to\N$ such that
for every graph~$G:=(V,E)$ of the class, any vertex~$v \in
V$, and any~$d \in \N$, the subgraph of~$G$ induced by all vertices
within distance at most~$d$ from~$v$ has treewidth at
most~$f(d)$.
\end{definition}

We refer readers to~\cite{FlumGrohe2006} for more details on
the notion of treewidth and locally bounded treewidth. The
crucial property of locally bounded treewidth graphs in our
context is that first-order queries can be answered in FPT time
on such graphs when the parameter is the size of first-order
formula~\cite[Chapter~12.2]{FlumGrohe2006}.

For ease of presentation we show how to find an equilibrium
where both players have support size~$\WLGsupportsize$ (the
algorithm can be easily adapted to support
sizes~$\WLGsupportsize_1,\WLGsupportsize_2 \leq
\WLGsupportsize$). Let~$I$ and~$J$ be two subsets
of~$\WLGsupportsize$ elements in~$[n]$. We say that two
matrices~$A^*,B^* \in \mathbb{Q}^{k \times k}$ \emph{occur} in~$\G$ at~$(I,J)$ if~$A^* =
A_{I,J}$ and~$B^* = B_{I,J}$. The pair~$(A^*,B^*)$ forms an
\emph{equilibrium pattern} if there exists an
equilibrium~$(x,y)$ where~$(A^*,B^*)$ occur in~$\G$
at~$(S(x),S(y))$. Our algorithm will try all
possible~$\ell^{2\WLGsupportsize^2}$ pairs of
matrices~$(A^*,B^*)$, and for each such pair it will determine
whether it is an equilibrium pattern.

When does a pair of matrices~$(A^*,B^*)$ form an equilibrium
pattern? The first obvious condition is that it occurs in~$\G$
at some pair of position sets~$(I,J)$. Furthermore, by
definition of a Nash equilibrium, there is a pair of
strategies~$(x,y)$ with~$S(x)=I$ and~$S(y)=J$, such that
neither player has a better alternative. The difficulty here
lies in the fact that, even given the support~$S(y)$ of the
column player, there may be too many possible strategies for
the row player that have supports different from~$I$. To
circumvent this, we define equivalence of rows with respect to
supports~$S(y)$, and of columns with respect to
supports~$S(x)$.

\begin{definition}\label{definition:equivalentstrategies}
Let~$I,J \subseteq [n]$. Two rows~$i_1,i_2 \in [n]$
are~\emph{$J$-equivalent} if~$A_{i_1,J}=A_{i_2,J}$. Similarly,
two columns~$j_1,j_2 \in [n]$ are~\emph{$I$-equivalent}
if~$B_{I,j_1}=B_{I,j_2}$.
\end{definition}

\begin{lemma}\label{lemma:violatingstrategies}
Let~$J$ be the support of the column player. For any row
strategy~$x$ there is a row strategy~$\hat{x}$ such that:
\begin{enumerate}[(i)]
\item the support~$S(\hat{x})$ contains at most one row
    from each~$J$-equivalence class
\item and for any column strategy~$y$ with support~$J$ we
    have~$\hat{x}^TAy = x^TAy$.
\end{enumerate}
The same is true for column strategies, given a support~$I$ of
the row player.
\end{lemma}

For each possible equilibrium pattern~$(A^*,B^*)$ we do the
following. For each choice of
rows~$A^\dagger\subseteq P^{1\times\WLGsupportsize}$ that
do not occur in~$A^*$ and each choice of
columns~$B^\dagger\subseteq P^{\WLGsupportsize\times 1}$
that do not occur in~$B^*$,
we create two matrices
\[
C=\left( \begin{array}{ll}
A^* & \mzero \\
A^\dagger & \mzero
\end{array} \right)\mbox{ and }
D=\left( \begin{array}{ll}
B^* & B^\dagger \\
\mzero & \mzero
\end{array} \right).
\]
We use Theorem~\ref{theorem:computingnash} to see if there is an equilibrium~$(x,y)$ in the game~$(C,D)$
with~$S(x)=S(y)=[\WLGsupportsize]$. If there is such an
equilibrium, then we proceed as follows to find an occurrence
of~$(A^*,B^*)$ that avoids the rows and columns which were not
chosen. For this let~$F_1$ be the rows which occur neither
in~$A^*$ nor in~$A^\dagger$ and let~$F_2$ be the columns which
occur neither in~$B^*$ nor in~$B^\dagger$. We say that~$F_1$
and~$F_2$ are \emph{forbidden} for~$(A^*,B^*)$. We note that
given~$(A^*,B^*)$, a set of rows~$F_1 \subseteq P^{1
\times \WLGsupportsize}$, and a set of columns~$F_2 \subseteq
P^{\WLGsupportsize \times 1}$, one can write a
first-order formula of size bounded by some function
in~$\WLGsupportsize$ and~$|P|=\ell$ to determine whether~$(A^*,B^*)$ has an
occurrence which avoids~$F_1$ and~$F_2$.
\begin{example}
Consider a win-lose game~$(A,B)$ encoded into relations~$A$ and~$B$
such that~$A(r,c)$ is true iff~$A_{r,c}=1$, and likewise for~$B$.
Then a $2 \times 2$ equilibrium where the pattern is two identity matrices can be found with the formula
\begin{align*}
\exists r_1,r_2,c_1,c_2 & A(r_1,c_1) \land \neg A(r_1,c_2) \land \neg A(r_2,c_1) \land A(r_2,c_2) \land \\
&B(r_1,c_1) \land \neg B(r_1,c_2) \land \neg B(r_2,c_1) \land B(r_2,c_2) \land \\
&\forall r' (\neg A(r',c_1)\lor \neg  A(r',c_2)) \land \forall c' (\neg B(r_1,c')\lor \neg B(r_2,c')).
\end{align*}
In general, with~$\ell$ different values, there would be~$\ell-1$ relations~$A_i$ and~$B_i$
encoding the game, where~$A_i(r,c)$ is true if~$A_{r,c}=z_i$,
for every~$z_i \in P$ except the zero value.
\end{example}
Since the number of
possible choices for~$F_1$ and~$F_2$ is bounded by some
function in~$\WLGsupportsize$ and~$\ell$, and for each such choice we can
determine whether~$F_1$ and~$F_2$ is a forbidden pair
for~$(A^*,B^*)$ in polynomial-time, the total time for
determining whether~$(A^*,B^*)$ is an equilibrium pattern is
FPT in~$\WLGsupportsize$ and~$\ell$. Since the number of pairs~$(A^*,B^*)$
is also bounded by a function in~$\WLGsupportsize$, the total
running time of our entire algorithm is also FPT
in~$\WLGsupportsize$ and~$\ell$.

To complete the proof of Theorem~\ref{theorem:localtreewidthgame}, let
us briefly argue completeness. Assume that there is any
equilibrium with support sizes equal to~$\WLGsupportsize$,
let~$I$ and~$J$ be the supports, and let~$A^*$ and~$B^*$ be
corresponding sub-matrices. Observe that we may set all entries
in columns outside~$J$ of~$A$ to zero without harm, ditto for
rows outside~$I$ in~$B$. According to
Lemma~\ref{lemma:violatingstrategies} it suffices to keep one
copy of each row outside~$A^*$ in~$A$ (also discard the
corresponding zero-row in~$B$ to keep the size the same). The
same is of course true for columns outside~$B^*$ in~$B$. Except
for a permutation this is equal to one of the games~$(C,D)$
that we considered. Therefore our algorithm will find such an
equilibrium if one exists.

\section{Unbalanced Games}
\label{section:unbalancedgames}

In this section we briefly consider~$k$-unbalanced bimatrix games. A
bimatrix game~$(A,B)$ is~\emph{$k$-unbalanced} if~$A,B \in
\mathbb{Q}^{k \times n}_{\geq 0}$ for
some~$k<<n$~\cite{KalyanaramanUmans2007,LiptonMarkakisMehta2003}
(i.e., the row player has a significantly smaller number of
strategies than the column player). We will show that a Nash
equilibrium can be computed in FPT-time with respect to~$k$
and~$\ell$, where~$\ell$ denotes the number of different
payoffs that the row player has, i.e.,~$\ell:=|\{A_{i,j} : 1
\leq i \leq k, 1 \leq j \leq n\}|$.

Similar to Definition~\ref{definition:equivalentstrategies} we
define two column strategies~$i,j\in[n]$ to be equivalent
if~$A_{*,i}=A_{*,j}$. (However, notice that unlike
Definition~\ref{definition:equivalentstrategies}, here
equivalence of column strategies is defined with respect to the
row player payoff.)

\begin{lemma}
\label{lemma:equivilanceequilibrium}
For each equilibrium there is an equilibrium where the column
player plays at most one column from each equivalence class.
\end{lemma}

Using Lemma~\ref{lemma:equivilanceequilibrium} we can easily
devise an FPT algorithm for computing a Nash equilibrium in our
setting. The algorithm simply guesses the support of the row
player and column player, and then uses
the method of Theorem~\ref{theorem:computingnash} to determine whether there exists
a Nash equilibrium corresponding to these sets of supports.
Observe that there are at most $\ell^k$ column-strategy
equivalence classes. Furthermore, according to Lipton,
Markakis, and Mehta~\cite{LiptonMarkakisMehta2003}, in
a~$k$-unbalanced game there always exists an equilibrium where
the column player has support size at most~$k+1$. Thus the
number of guesses the algorithm makes is bounded by~$2^k \cdot
\binom{\ell^k}{k+1}=\ell^{O(k^2)}$, and for each such guess,
the amount of time required is polynomial. This completes the
proof of Theorem~\ref{theorem:unbalancedgame}.

\bibliographystyle{plain}
\bibliography{biblio}

\newpage
\appendix

\section{Omitted proofs}

\subsection{Section~\ref{section:localtreewidthgame}}

\begin{proof}[Lemma~\ref{lemma:violatingstrategies}]
We prove the lemma row strategies given a support~$J$ of the column player. The proof for column strategies is similar. Let~$x$ be a row strategy which includes two strategies~$i_1$ and~$i_2$ that are~$J$-equivalent. It suffices to show that there is a row strategy~$\hat{x}$ with~$S(\hat{x}) = S(x) \setminus {i_2}$ and~$\hat{x}^TAy=x^TAy$ for any column strategy~$y$ with support~$J$. For this, take~$\hat{x}$ to be the strategy defined by~$\hat{x}_{i_1}:=x_{i_1}+x_{i_2}$,~$\hat{x}_{i_2}:=0$, and~$\hat{x}_i =x_i$ for all~$i\in \{1,\ldots,n\} \setminus \{i_1,i_2\}$. Let~$y$ be any column strategy with~$S(y) = J$. By definition of~$J$-equivalence, we have~$(Ay)_{i_1} = (Ay)_{i_2}$. Thus, we get that
\[
(\hat{x}^TAy)_{i_1} + (\hat{x}^TAy)_{i_2} = (x^TAy)_{i_1} + (x^TAy)_{i_2}.
\]
Furthermore, since~$\hat{x}$ and~$x$ equal on all entries~$i \in [n] \setminus \{i_1,i_2\}$, we know that
\[
\sum_{i \neq i_1,i_2} (\hat{x}^TAy)_i = \sum_{i \neq
i_1,i_2} (x^TAy)_i,
\]
and thus
\begin{align*}
(\hat{x}^TAy) & = \sum_{i \neq i_1,i_2} (\hat{x}^TAy)_i+(\hat{x}^TAy)_{i_1} + (\hat{x}^TAy)_{i_2} \\
& = \sum_{i \neq i_1,i_2} (x^TAy)_i + (x^TAy)_{i_1} + (x^TAy)_{i_2} \quad= (x^TAy).
\end{align*}
Therefore~$\hat{x}^TAy = x^TAy$ for all column strategies~$y$ with~$S(y)=J$. \qed
\end{proof}

\subsection{Section~\ref{section:unbalancedgames}}

\begin{proof}[Lemma~\ref{lemma:equivilanceequilibrium}]
Let~$(x,y)$ be an equilibrium, and suppose that~$S(y)$ includes two equivalent strategies~$i$ and~$j$. To prove the lemma it suffices to show that the there exists an equilibrium~$(x,y^*)$ with~$S(y^*) = S(y) \setminus \{j\}$. For this, let us take~$y^*$ to be the strategy vector defined by~$y^*_i:=y_i+y_j$,~$y^*_j:=0$, and~$y^*_x =y_x$ for all~$x \in \{1,\ldots,n\} \setminus \{i,j\}$. Clearly~$S(y^*) = S(y) \setminus \{j\}$.
Thus, for each~$s \in S(y^*)$ we have
$$
(x^{T}B)_{s} \geq (x^{T}B)_{j}, \forall j \neq s,
$$
since this condition holds for all~$s \in S(y) \supseteq S(y^*)$ by Lemma~\ref{lemma:nashcondition}.
Furthermore, due to the definition of equivalence, a simple calculation shows that~$(Ay)_{s} = (Ay^*)_s$ for all~$s \in \{1,\ldots,k\}$. Therefore, for each~$s \in S(x)$, we get again using Lemma~\ref{lemma:nashcondition} that
$$
(Ay^*)_{s} = (Ay)_{s} \geq (Ay)_{j}, \forall j \neq s.
$$
It follows that~$(x,y^*)$ satisfies both conditions of Lemma~\ref{lemma:nashcondition}, and so it is indeed an equilibrium. \qed
\end{proof}

\end{document}